\documentclass[12pt,a4paper,times]{article}
\usepackage[utf8]{inputenc}
\usepackage{amssymb,amsmath,amsthm,relsize,bm}
\numberwithin{equation}{section}
\usepackage[a4paper,total={6in,9in}]{geometry}
\newcommand{\ts}{\textsuperscript}
\usepackage[english]{babel}
\usepackage{graphicx,caption,float,empheq,blindtext,enumerate,subcaption}
\usepackage{mathptmx}
\usepackage{url}

\usepackage{amsthm, cite}
\usepackage{authblk}
\usepackage{url}
\usepackage{enumitem}

\newcommand{\me}{\mathrm{e}}
\DeclarePairedDelimiter\abs{\lvert}{\rvert}%
\DeclarePairedDelimiter\norm{\lVert}{\rVert}%

\makeatletter
\let\oldabs\abs
\def\abs{\@ifstar{\oldabs}{\oldabs*}}
\let\oldnorm\norm
\def\norm{\@ifstar{\oldnorm}{\oldnorm*}}
\makeatother

\usepackage[bottom]{footmisc}
\usepackage{ragged2e}
\makeatletter
\newcommand\HUGE{\@setfontsize\Huge{20}{46}}
\makeatother 
\makeatletter
\newcommand\BIG{\@setfontsize\Large{16}{16}}
\makeatother

\newtheorem{theorem}{Theorem}[section]

\title{
\HUGE Dynamic Network $3$ -- $0$ FIFA Rankings: \\
\vspace{0.5cm}
\BIG Replacing an inaccurate, biased, and exploitable ranking system
\vspace{2.5cm}
}
\normalsize
\author[1,2]{Sam Abernethy}
\affil[1]{Mathematical Institute, University of Oxford, Oxford, UK, OX2 6GG}
\affil[2]{Applied Physics Department, Stanford University, Stanford, USA, 94305}
\begin{document}
\begin{titlepage}
\maketitle

\vspace{\fill}
\centering

\begin{abstract}
We explore the advantages of representing international football results as a directed network in order to give each team a rank. Two network-based models --- Static and Dynamic --- are constructed and compared with the FIFA Rankings. The Dynamic Model outperforms the FIFA Rankings in terms of World Cup predictive accuracy, while also removing continental bias and reducing the vulnerability of the FIFA Rankings to exploitation.

%International football results are interpreted as a directed network to give each team a rank. Two network-based models --- Static and Dynamic --- are constructed and compared with the current FIFA Rankings. The Dynamic Model outperforms the FIFA Rankings in terms of World Cup predictive accuracy, while also removing continental bias and reducing the vulnerability of the FIFA Rankings to exploitation. This is done using two intuitive parameters --- relative importance of indirect wins and bygone wins --- whose values are determined to be $0.07$ and $0.9$ respectively. We suggest that FIFA should implement this model as its future ranking system.
\end{abstract}

\pagenumbering{gobble}
\end{titlepage}

\newpage
%\lhead{}
%\lfoot{}
%\cfoot{\thepage}
%\rfoot{}

\pagenumbering{arabic}
\setcounter{page}{1}

\section{Introduction}

For a wide variety of sports, the central organising body ranks the participating teams based on their results. These ranking systems serve two primary purposes: to reward teams for past successes, and to provide predictions about the relative quality of the teams. In this work, we address only one sport --- association football or soccer --- at the international level. The governing body for international football is the F\'{e}d\'{e}ration Internationale de Football Association (FIFA), who release monthly rankings of each team. These rankings will serve as the benchmark for this work, and will be denoted the `FIFA Rankings.' 

Ranking systems measure teams' past success, but they can also influence teams' future success. For major tournaments such as the FIFA World Cup and continental championships, teams are divided into pools based on their rank. Having a higher rank means that you will be in a pool with, and play against, teams who are weaker on average \cite{Guyon}. As such, ranking systems can influence which teams progress to the knockout stages of tournaments --- earning millions of dollars in the process. For that reason, among others, one hopes that the FIFA Rankings are fair, unbiased, robust, rewarding of past success, and predictive of future results. However, as we will explain in Section \ref{sec:bad}, this is not the case.

There has been extensive research to improve ranking systems such as FIFA's \cite{stef, dan, moon, langville2}. In terms of football ranking systems, the main division is between those that predict goals scored and conceded, and those that predict the win-loss outcome directly \cite{Bootsma}. We focus our attention on the latter category. Approaches for predicting the win-loss outcome range from Elo models wherein teams exchange points after every match \cite{Kovalchik2016, Mangan2016, Xiong} to Google's PageRank algorithm, adapted from measuring the importance of websites, wherein rankings are derived from a network \cite{Chartier2011, langville, Lazova2015, Radicchi2011}, to more sophisticated random forest and physically-inspired models \cite{Bacco, Groll2018, Criado2013}.
%Seemingly the most efficient methodology for win-loss predictive success is the family of Elo models, first introduced by Arpad Elo in the 1960's.
%\footnote{See Appendix \ref{app:elo} for the methodology of this model}
%Elo models have been applied to many sports and games including chess, gaelic football \cite{Mangan2016}, tennis \cite{Kovalchik2016}, and association football \cite{Xiong}. 
It has been proven that some Elo models, with finely tuned parameters, outperform the FIFA Rankings in predicting the outcome of football matches \cite{Lasek2013}. However, these models often involve unintuitive parameters with no real-world meaning, leading to an enigmatic system that is not meaningful to teams and fans alike. It is important to stress that the problem of ranking in sports finds many connections with the problem of defining importance in networks, where a variety of centrality measures have been developed based on different principles.

The main purpose of this paper is to show that network-based methods provide a competitive alternative to the FIFA Rankings due to their simplicity, interpretability, and high predictive power. To do so, we consider the concept of indirect wins, inspired by recursive centrality measures like eigenvector centrality and Katz centrality in network science \cite{Katz}.
%and exploited in sports ranking systems in American football, for example \cite{Newman2005}.
After constructing a Static Model based on related work for American football \cite{Newman2005, Keener}, we extend it by incorporating the temporal dimension of the games. Our Dynamic Model is constructed by placing less importance on wins from long in the past, as in \cite{Motegi2012}. These are referred to as bygone wins. The two models, which can be viewed as representative examples of network-based approaches, are then tested on the largest sporting event in the world: the World Cup, held every four years.

%In this paper, our goal is to develop a network-based model --- using intuitive parameters with real-world comprehensibility --- that has higher predictive accuracy than the FIFA Rankings. To do so, we use the concept of indirect wins (first introduced for American football \cite{Newman2005}) to construct a Static Model. Extending the original model to incorporate temporal network changes, we introduce the concept of bygone wins --- placing less importance on wins from long in the past, as in \cite{Motegi2012} --- to construct a Dynamic Model. Our models will focus on the largest sporting event in the world: the World Cup, held every four years.

The structure of this work is as follows. In Section \ref{sec:fifa} we describe and analyse the flaws of the FIFA Rankings. Based on these flaws, we develop two network-based models --- Static (Section \ref{sec:static}) and Dynamic (Section \ref{sec:dynamic}) --- and compare them with the FIFA Rankings. The implications of these models are analysed in Section \ref{sec:pred} and discussed in Section \ref{sec:disc}.

\section{FIFA Rankings}\label{sec:fifa}
\subsection{Methodology}

In the FIFA Rankings, each national team has a point score which determines its rank: the more points, the higher the rank. Using the information on FIFA's website \cite{fifa} we will describe the methodology for calculating a team's points.

The number of \textbf{P}oints gained in a match is determined by multiplying factors for the \textbf{M}atch result, the match \textbf{I}mportance, the \textbf{C}ontinental strength, and the opposing \textbf{T}eam: 
\begin{equation}
\bm{P} = \bm{M} \times \bm{I} \times \bm{C} \times \bm{T}.
\end{equation}

For the \textbf{M}atch result, teams gain $3$ points for a win, $1$ point for a draw, and $0$ points for a loss. If the match is decided in a penalty shoot-out, the winner gains $2$ points while the loser gains $1$ point.

The match \textbf{I}mportance is equal to $4$ for World Cup matches, $3$ for continental tournaments, $2.5$ for World Cup or continental qualifiers, and $1$ for friendly matches.

The \textbf{C}ontinental strength factor is the average of the two teams' continental strengths: $1$ for South America, $0.99$ for Europe, and $0.85$ for all other continents.

Finally, the opposing \textbf{T}eam factor is given by $(200 - R_o)$, where $R_o$ is the rank of the opponent in the most recent FIFA Ranking (with a minimum of $50$).

For each team, the points gained in all matches within the past year are \textit{averaged} to give $P_{Y_1}$. Similarly, they are averaged for months $12$--$24$ to give $P_{Y_2}$ and so on for $P_{Y_3}$ and $P_{Y_4}$. The total points used in the FIFA Rankings come from a block-wise depreciating average calculated over a moving four-year window:
\begin{equation}\label{eq:fifa}
P_{\text{FIFA}} =  P_{Y_1} +0.5 \times P_{Y_2} + 0.3 \times P_{Y_3} + 0.2 \times P_{Y_4}.
\end{equation}

In the first week of each October preceding a World Cup, the current FIFA Ranking is used to draw the World Cup pools from the teams that have qualified. To standardise our analysis, we will use this same time when examining the network-based ranking systems throughout this work.

The simplicity of this ranking system, which has been in place since 2006, was largely motivated by widespread criticism received for earlier iterations. The previous system (used from 1999 until 2006) had many more parameters and the following arguments against the modern system would likely hold against previous systems as well.

\subsection{Critiques}\label{sec:bad}

There are three main critiques of the FIFA Rankings: its predictive accuracy is low, it is biased towards certain continents, and it can be exploited by careful selection of which teams should be played when. We will address these three issues in turn, aiming to demonstrate that the FIFA Rankings are inadequate and should be replaced.

We take the metric of predictive accuracy to be the percentage of matches at a World Cup that were predicted by the rankings from the preceding October, when the pools were drawn. Although there are many other metrics which could be chosen (binomial deviance or mean squared error to name a few), we take this value to be readily understandable and a good indicator of a ranking system's predictive accuracy: purely how many matches it predicts. We consider only matches which have a winner, discarding group stage matches which end in a draw. For this reason, predictive accuracy is the percentage of matches (with a winner) whose winner was the higher ranked team. For the World Cups from 1994 to 2002, the FIFA Rankings predicted $69\% \pm 3\%$ of the wins \cite{statslife}. In the past three World Cups, their accuracy has risen slightly to $71\% \pm 11\%$. Although this accuracy is above $50\%$, it is not as high as other methods such as Elo \cite{Lasek2013}.

The second critique of the FIFA Rankings is that they are biased based on continent. Consider a simple scenario in which Australia and Chile have the same rank, as they nearly do in 2018. If they each beat Brazil in a friendly, Australia would receive $92.5\%$ of the points that Chile receives --- for the exact same outcome. (This is because the continental strength factor, averaged between the two teams in a match, is $1$ for Chile and Brazil and only $0.85$ for Australia.) This is a continental bias because the relative strengths of the teams are double counted: the opposing team factor \textit{and} the continental strength factor means that teams from traditionally weak continents are disadvantaged. Although FIFA likely did this to maximise their predictive accuracy --- given the poor model they chose to construct --- it results in a clear continental bias.

%The numerous arbitrary constants therefore lead to a biased ranking system.
%contain many arbitrary constants with no real-world significance. Specifically, we consider the following constants to be arbitrary: $\bm T$ ($200 - R_o$), $\bm C$ ($1$, $0.99$, or $0.85$), $\bm I$ ($4$, $3$, $2.5$, or $1$), and the block-wise depreciating average ($0.5$ for $1$--$2$ years in the past, etc.). The extent to which these constants affect the rankings remains unclear. 
%However, a simple example demonstrates how the FIFA Rankings are indeed biased based on continent. Assume that Australia and Chile were tied in the FIFA Ranking, as they nearly are in 2018. If they each beat Brazil in a friendly, Australia would receive $\frac{0.85+1.0}{2}= 0.925$ times the number of points that Chile receives --- only $92.5\%$ of the reward for the exact same outcome! The numerous arbitrary constants therefore lead to a biased ranking system.

The final critique that we will address is the fact that the FIFA Rankings are exploitable. Due to the averaging of point scores, teams that play few friendly matches currently stand to gain the most from the current system. This fact negatively impacts teams such as England that refuse to alter their schedule to suit the FIFA Rankings \cite{telegraph}. A summary analysis of the optimal matches for teams to play has been released by Lasek et al. \cite{Lasek2016}. They argue, for example, that Ukraine could have been in a higher pool for the 2016 European championship had Ukraine followed their proposed methodology for expected point maximisation. 
%As a result, it has been widely reported that FIFA will reevaluate their ranking system after the 2018 World Cup \cite{telegraph}. %Furthermore, they concluded that the FIFA Rankings unfairly punish teams for hosting major tournaments since they play only friendlies instead of qualification matches. 
%Since FIFA is clearly in a position where a major change to the ranking system would be welcome, we will now introduce an alternative system based on networks.
%For this reason, we will now  network-based model that could be used as an alternative to the FIFA Rankings.

\subsection{How can it be improved?}

Our motivation for adopting a network-based model for international football is threefold: improve predictive accuracy, remove the continental bias, and develop a more robust ranking system that is not exploitable. 

A crucial source of data, currently ignored by FIFA, is the concept of indirect wins (as in \cite{Keener, Newman2005} based on the work of Katz \cite{Katz}). The classic argument goes as follows: ``my team A beat team B, who in turn beat your team C --- therefore, my team is better than your team." This argument makes intuitive sense and also increases predictive accuracy when implemented in similar ranking systems for wrestling \cite{Bigsby2017}, tennis \cite{Motegi2012}, and American football \cite{Newman2005} so we choose to incorporate it as the central tenet of our models.

The other two goals are also addressed using this methodology. The team and continent multipliers will be replaced by the concept of indirect wins; beating a strong team gives you many indirect wins, which removes the need to scale points based on the opposition or their continent. Furthermore, we will no longer average the point scores, thereby removing the ability to exploit the system by playing as few friendly matches as possible.

To achieve these goals, we will construct two network-based models --- Static and Dynamic --- with intuitive parameters chosen based on predictive accuracy. We aim to make the ranking algorithms readily understandable, unbiased, and not exploitable, \textit{before} optimizing their predictive power. There is a tradeoff between simplicity and predictive accuracy that we hope to mitigate by presenting two models --- the Dynamic Model will be more complex and, we hypothesise, consequently more accurate in its predictions. 

%2006: 60.3\% \cite{Suzuki2008}. Three previous: 67.6\%. Thus an average of 65.7\%.
%2006: 60.4\% . Three previous: 68.9\%. 2010: 72.5\%
% https://www.statslife.org.uk/sports/1530-how-well-do-fifa-s-ratings-predict-world-cup-success

\section{Static Model}\label{sec:static}
\subsection{Derivation}

To construct a network-based model for international football, we must first choose the data of importance. In a world brimming with massive quantities of data, there are clearly many unimportant factors we could include --- did Lionel Messi play $72$ minutes or more in his last away match in a country in the Northern Hemisphere?  
%However, we desire a model which rewards past success and is not purely driven by predictive accuracy. 
As such, we make a significant choice, as in \cite{Newman2005} and Keener's direct approach \cite{Keener}: take only the win-draw-loss result for each match, rather than the goals scored or conceded. This rewards teams based on the most important aspect of a match --- the result --- rather than incentivising running up the score. 

The results from $37,653$ international matches dating back to $1872$ are taken from \cite{kaggle}. Table \ref{tab:kaggle} shows example data for the first and last match used in this work.
\begin{table}[H]
\caption{Example data used for each international match.}\label{tab:kaggle}
\vspace*{-5mm}
\begin{center}
\begin{tabular}{|c|c|c|c|c|c|}
\hline
Date & Team 1 & Team 2 & Outcome & Type \\
\hline
1872-03-08 & Scotland & England & Scotland Win & Friendly \\
\vdots &\vdots &\vdots &\vdots &\vdots \\
2017-09-05 & Egypt & Uganda & Egypt Win & WC qualifier \\
\hline 
\end{tabular}
\end{center}
\end{table}

The results from Table \ref{tab:kaggle} can be represented as a directed network, where team A defeating team B is modelled by a directed edge pointing from team A's node to team B's node. Such a network can be expressed in terms of an adjacency matrix $A \in \mathbb{R}^{n \times n}$ where $n \approx 220$ is the number of national teams. (Note that $n$ varies as some teams, such as North and South Vietnam, cease to exist.) 

Based on work done for American football rankings \cite{Keener, Newman2005}, we define the adjacency matrix as follows: $A_{ij}$ is the number of times team $j$ has defeated team $i$ in a given time period. In American football there are no draws, so in our model we interpret a draw as half of a win and half of a loss --- $A_{ij}$ and $A_{ji}$ are incremented by $0.5$ for each draw between team $i$ and team $j$ \cite{Lasek2013, Keener}. The number of wins for team $i$ is denoted their \textit{direct wins} (noting that includes the contributions of draws) and is equal to $\sum_j A_{ji}$ --- simply adding up all of the elements in the $i$\ts{th} column of $A$.

The number of indirect wins for team $i$ at distance $2$ (i.e., when team $i$ beat team $j$ who beat team $k$) can be written as $\sum_{kj} A_{kj} A_{ji}$.  
%\footnote{The shorthand notation $\sum_{abc \cdots} = \sum_{a=1}^n \sum_{b=1}^n \sum_{c=1}^n \cdots$ will be used from here on.} 
Similarly, we can write the number of indirect wins for team $i$ at distance $3$ as $\sum_{hkj} A_{hk} A_{kj} A_{ji}$, and so on, as outlined in Equation \eqref{eq:fifa}. 

Indirect wins at a distance $d$ are discounted by a multiplicative factor $\alpha^{d-1}$, where $\alpha<1$ is a free parameter called the \textit{indirect win factor}. From this construction, Newman and Park define a win score $w_i$ for team $i$:
\begin{align*}
w_i &= \sum_j A_{ji} + \alpha \sum_{kj} A_{kj} A_{ji} + \alpha^2 \sum_{hkj} A_{hk} A_{kj} A_{ji} + \cdots \\
&= \sum_j \left(1 + \alpha \sum_{k} A_{kj} + \alpha^2 \sum_{hk} A_{hk} A_{kj} + \cdots \right) A_{ji} \\
&= \sum_j \left(1 + \alpha \left[ \sum_{k} A_{kj} + \alpha \sum_{hk} A_{hk} A_{kj} + \cdots \right] \right) A_{ji} \\
&= \sum_j \left( 1 + \alpha w_j \right) A_{ji}.
\end{align*}

The loss score $l_i$ is defined in a similar way:
\begin{align*}
l_i &= \sum_j A_{ij} + \alpha \sum_{jk} A_{ij} A_{jk} + \alpha^2 \sum_{jkh} A_{ij} A_{jk} A_{kh} + \cdots \\
&= \sum_j A_{ij} (1 + \alpha l_j),
\end{align*}
using the same procedure as for the win score. (Note that we require the same condition on $\alpha$ for this sum to converge as well.)

As in many sports, we care about the difference between the win score and the loss score. We thus define the total score as $s_i = w_i - l_i$. The total score forms the basis of Newman and Park's ranking system \cite{Newman2005}. 

To vectorise the scores, we note that the win and loss scores can be rewritten as
\begin{align*}
\bm{w} = {\bm{k}}^{\text{out}} + \alpha A^T {\bm{w}} & \quad \implies \quad \bm{w} = \left( I - \alpha A^T \right)^{-1} \bm{k}^{\text{out}},  \\ 
\bm{l} = {\bm{k}}^{\text{in}} + \alpha A \bm{l} & \quad \implies \quad \bm{l} = \left( I - \alpha A \right)^{-1} \bm{k}^{\text{in}},
\end{align*}
where $\bm{w} = (w_1, w_2,\ldots)$, $\bm{l} = (l_1,l_2,\ldots)$, $\bm{k}^{\text{out}} = (\sum_j A_{j1}, \sum_j A_{j2}, \ldots)$, $\bm{k}^{\text{in}}=(\sum_j A_{1j}, \sum_j A_{2j}, \ldots)$, and $I$ is the $n \times n$ identity matrix. The vectors $\bm{k}^{\text{out}}$ and $\bm{k}^{\text{in}}$ are simply the vectorised versions of the direct wins and direct losses respectively.

In Newman and Park's model, each American football match had equal importance since they were part of a discrete season. However, this is not the case for international football. For that reason, our Static Model extends the Newman and Park model by allowing the contribution to $A_{ij}$ for each match to be multiplied by two factors: match importance and time past. We follow the FIFA convention and take the importance to be $4$ for World Cup matches, etc., as outlined in Section \ref{sec:fifa}. Since this network-based model no longer averages over the number of matches played (biasing results towards teams that play fewer friendly matches) this choice appears to be a reasonable method to assign importance where it is deserved. To make the static model mirror the FIFA Rankings (thus hopefully making it more appealing to FIFA for an easy change) we also use their block-wise depreciating factor: the contribution to $A_{ij}$ is multiplied by $1$ if it took place within the past year, $0.5$ if it was within two years, and so on.

Although left unproven in the previous two applications of this model \cite{Motegi2012,Newman2005}, we will show that $\alpha$ is limited to $\alpha \in [0,\ \lambda_{\text{max}}^{-1}]$, where $\lambda_{\text{max}}$ is the spectral radius of $A$. This limitation comes from interpreting the win and loss scores as generalisations of Katz centrality \cite{Katz} wherein Katz's `probability of effectiveness' parameter is also limited by the spectral radius of the matrix in question. The following Theorem, modified from \cite{Mic2009} to suit our purpose, shows why this is the case for the win score. The corresponding proof for the loss score is almost identical.
%We will consider only $\alpha \in [0,\ \lambda_{\text{max}}^{-1}]$, where $\lambda_{\text{max}}$ is the spectral radius of $A$. 
%To guarantee convergence of this sum, $\alpha$ must be less than $\lambda_{\text{max}}^{-1}$, where $\lambda_{\text{max}}$ is the spectral radius of $A$ \cite{Newman2005}. Therefore, we will consider only $\alpha \in [0,\ \lambda_{\text{max}}^{-1}]$. 
%The following Theorem shows the reason for this limitation.
\begin{theorem}
The win score $w_i$ converges if and only if $\alpha < \lambda_{\text{max}}^{-1}$. 
\end{theorem}
\begin{proof}

Consider the matrix $A^k$. If $k=1$, then the $i$\ts{th} column sum of $A^1$ is the number of direct wins for team $i$. Similarly, the $i$\ts{th} column sum of $A^2$ is the number of indirect wins for team $i$ at distance $2$, and so on. For this reason, we define the matrix 
$$T = A + \alpha A^2 + \alpha^2 A^3 + \cdots,$$ 
with free parameter $\alpha$. The $i$\ts{th} column sum of $T$ will therefore be the win score for team $i$. We can rewrite $T$ as follows:
\begin{align*}
\alpha T &= (\alpha A) + (\alpha A)^2 + (\alpha A)^3 + \cdots \\ 
&= \sum_{k=1}^{\infty} (\alpha A)^k.
\end{align*}
This sum converges if and only if $\abs{\alpha \lambda_{\text{max}}} < 1$, since it involves a geometric series of matrices. By the Perron-Frobenius Theorem for non-negative matrices, $\lambda_{\text{max}} \geq 0$ (as $A$ is irreducible). Since we have defined $\alpha \geq 0$ we may remove the absolute value signs. Therefore, $T$ (and hence the win score) converges if and only if $\alpha < \lambda_{\text{max}}^{-1}$.

% Lasek says 'power series converges whenever alpha < lambdamax-1
% newman says 'power series converges only if alpha < lambdamax-1
% motegi says 'if it exceeds, then power series diverges'

%First, we define $\abs{A_i} = \sum_j A_{ji}$. In this proof we will exclusively use the $1$-norm.
%\begin{align*}
%\abs{w_i} &= \abs{\sum_j A_{ji} + \alpha \sum_{kj} A_{kj} A_{ji} + \alpha^2 \sum_{hkj} A_{hk} A_{kj} A_{ji} + \cdots} &  \\
%&\leq \abs{A_i} + \alpha \abs{A A_i} + \alpha^2 \abs{A^2 A_i} + \cdots & (\text{since } \alpha\geq0) \\
%&\leq \abs{A_i} + \alpha \norm{A} \abs{A_i} + \alpha^2 \norm{A^2} \abs{A_i} + \cdots & (\text{matrix multiplication property}) \\
%&\leq \abs{A_i} \left(1 + \sum_{k=1}^{\infty} \alpha^k \norm{A}^k  \right) & (\text{since } \norm{A^k} \leq \norm{A}^k).  \\
%\end{align*}

%Now we use the spectral radius property that $\lambda_{\text{max}} \leq \norm{A^k}^{1/k}$ to deduce a limitation on $\lambda_{\text{max}}$: $\lambda_{\text{max}}^k \leq \norm{A^k} \leq \norm{A}^k$.

%If $\alpha \lambda_{\text{max}} < 1$ then 

%TOTALLY WRONG AS THE SPECTRAL RADIUS GOES THE OTHER WAY!!!
%Therefore, the win score series converges if $ \sum_{k=1}^{\infty} \alpha^k \lambda_{\text{max}}^k$ converges. Since $\alpha$ and $\lambda_{\text{max}}$ are both non-negative, this is true when $\alpha \lambda_{\text{max}}<1$, or equivalently $\alpha <\lambda_{\text{max}}^{-1}$. This proves that the win score series converges if $\alpha <\lambda_{\text{max}}^{-1}$.
\end{proof}

To summarise, the Static Model has replaced the continental and opposing team multipliers with one parameter, $\alpha \in [0,\ \lambda_{\text{max}}^{-1}]$, which is a measure of the relative importance of indirect wins. 
%Newman and Park examined the retrodictive accuracy of their model and found it to be roughly $80\%$ for a wide range of $\alpha$ values and across many seasons \cite{Newman2005}. (Retrodictive accuracy is the percentage of matches which are won by the ultimately higher-ranked team --- essentially the predictive accuracy about the past.) 
Newman and Park's elegant mathematical framework for ranking sports teams is easily generalisable --- to other sports, or to incorporate margin of victory or home field advantage, for example. However, one of the most attractive qualities of this model is the simplicity of having only one free parameter, $\alpha$. For this reason, we hypothesise that the Static Model is a candidate to replace the FIFA Rankings. 

\subsection{Parameter Estimation}

To support this claim, we compare the predictions of the Static Model to those of FIFA. Using the Static Model, we construct rankings for each of the past three Octobers preceding a World Cup (when the pools were selected). As shown in Figure \ref{fig:stat}, the predictive accuracy of the Static Model is comparable to FIFA for $\alpha \geq 0.7 \times \lambda_{\text{max}}^{-1}$ but on average it is lower by roughly $5\%$.   
\begin{figure}[H]
\begin{center}
\includegraphics[width=0.75\textwidth]{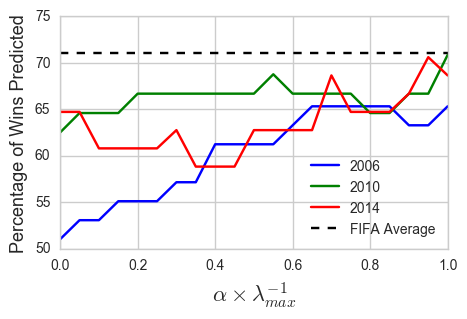}
\end{center}
\vspace*{-5mm}
\caption{Comparison of Static Model with FIFA Rankings for the past three World Cups. Sample $\alpha$ values were taken in the range $[0,\ \lambda_{\text{max}}^{-1}]$ with a step size of $0.05 \times \lambda_{\text{max}}^{-1}$.}\label{fig:stat}
\end{figure}

%It was shown \cite{Lasek2016} that the retrodictive accuracy of a similar network-based model (without the scaling factors for match importance or time past) was lower than that of FIFA. 
The question remains, however, which $\alpha$ value should be selected. As can be seen in Table \ref{tab:lamb}, the values for $\lambda_{\text{max}}^{-1}$ are reasonably consistent and generally fall in the range $0.05$--$0.06$. The fluctuations motivate our definition of $\alpha$ as a fraction of $\lambda_{\text{max}}^{-1}$ rather than as a fixed constant, in case future adjacency matrices differ substantially from those in the past.

\begin{table}[H]
\caption{$\lambda_{\text{max}}^{-1}$ values for the adjacency matrices preceding each World Cup.}\label{tab:lamb}
\vspace*{-5mm}
\begin{center}
\begin{tabular}{|c|c|c|}
\hline
World Cup & $\lambda_{\text{max}}$ & $\lambda_{\text{max}}^{-1}$ \\
\hline
2006 & 20.9 & 0.048 \\
\hline
2010 & 18.8 & 0.053 \\
\hline
2014 & 17.6 & 0.057 \\
\hline
%2018 & 15.8 & 0.063 \\
%\hline
\end{tabular}
\end{center}
\end{table}

Newman and Park selected their $\alpha$ based on its ability to predict the results from the past season, denoted `retrodictive accuracy.' Using this approach on our data, it was determined that the $\alpha$ value which maximised retrodictive accuracy was $0$. However, maximising retrodictive accuracy fails to capture a meaningful ranking of the teams for predictive purposes, as can be seen by the low predictive accuracy in Figure \ref{fig:stat} when $\alpha=0$. We suggest that a reasonable value for $\alpha$ is $0.95 \times \lambda_{\text{max}}^{-1}$. However, we note that the predictive accuracy is relatively insensitive to varying $\alpha$ in the range $0.7$--$1.0$ ($\times \lambda_{\text{max}}^{-1}$). 

For Newman and Park's application to American football, the rankings were specific to each discrete season. This meant that they required an absolute value for $\alpha$ rather than having it defined in terms of $\lambda_{\text{max}}$. In our model, however, we can use a sliding four-year window for the calculation of $\lambda_{\text{max}}$. The $\alpha$ value that maximised their retrodictive accuracy, $0.8\times \lambda_{\text{max}}^{-1}$, falls into the range for our model that is relatively insensitive to $\alpha$ variation \cite{Newman2005}. Since Newman and Park chose their $\alpha$ based on Monte Carlo simulations and retrodictive accuracy maximising heuristics, we feel justified in taking $\alpha = 0.95 \times \lambda_{\text{max}}^{-1}$ as the single parameter for the Static Model. (The limitations of this approach will be discussed in Section \ref{sec:disc}.)

\section{Dynamic Model}\label{sec:dynamic}
\subsection{Derivation}

Recently, there has been a developing body of work that suggests dynamic network-based ranking systems are preferable to their static counterparts \cite{Holme2012,Liao2017}. In such a temporal network, information about \textit{when} an event occurred can also contribute meaningfully to the data. The quality of football teams fluctuates over time, and we believe that defeating a team at their peak should be worth more than defeating them when they are in a slump. We also believe that wins yesterday should be rewarded more than wins from 364 days ago. For these reasons, we will now construct a network-based Dynamic Model following the notation from a similar application to tennis by Motegi and Masuda \cite{Motegi2012}. We desire a model which has increased predictive accuracy while maintaining the simplicity and intuitive structure of the Static Model. To do so, we will develop a model which discounts indirect wins and bygone wins in a symmetric way.

Let $A_{t_n}$ be the adjacency matrix for the matches played at $t_n$, as in the Static Model. The resolution of $t_n$ is equal to one day, meaning that each adjacency matrix is very sparse; only a handful of matches are played on each day, if any. The vectorised dynamic win score at time $t_n$ is defined as 
\begin{equation}
\bm{w}_{t_n} = W_{t_n}^T \bm{1},
\end{equation}
where 
\begin{align}\label{eq:dyn}
W_{t_n} =  A_{t_n} &+ \me^{-\beta (t_n - t_{n-1})} \mathlarger{\sum}_{m_n \in \{0,1\}} \alpha^{m_n} A_{t_{n-1}} A_{t_n}^{m_n} \nonumber \\
& + \left( \me^{-\beta (t_n - t_{n-2})} \mathlarger{\sum}_{m_{n-1},~ m_n \in \{0,1\}} \alpha^{m_{n-1}+m_n} A_{t_{n-2}} A_{t_{n-1}}^{m_{n-1}} A_{t_n}^{m_n} \right) \\
& + \cdots + \me^{-\beta (t_n - t_{1})} 
\mathlarger{\sum}_{m_{2},~ \ldots,~ m_n \in \{0,1\}} \alpha^{\sum_{i=2}^n  m_i} A_{t_{1}} A_{t_{2}}^{m_{2}} \cdots A_{t_n}^{m_n}. \nonumber
\end{align}

The dynamic win score is significantly more complex than its static counterpart due to the discounting factors for both indirect wins and bygone wins. The first term after $A_{t_n}$ represents the contribution to the win score from indirect wins that terminated (or direct wins that occurred) at $t_{n-1}$. To be precise, $m_n=0$ represents direct wins (hence $\alpha$ is raised to the power of $0$) whereas $m_n=1$ represents indirect wins at distance $2$ (hence $\alpha$ is raised to the power of $1$). For this reason, $m_i$ can be interpreted as a binary indicator of whether the (indirect) win involves a match at time $t_i$. To illustrate how this formula captures each possible time and distance of indirect victory, Table \ref{tab:four} summarises the weights and meanings of the four contributions to the term shown in brackets in Equation \eqref{eq:dyn}. These four contributions represent the contribution to the win score from indirect wins that terminated (or direct wins that occurred) at $t_{n-2}$; each term thereafter follows the same logic.
\begin{table}[H]
\caption{Summary of contributions to Equation \eqref{eq:dyn} for different $m_{n-1}$ and $m_n$ values.}\label{tab:four}
\vspace*{-5mm}
\begin{center}
\begin{tabular}{|c|c|c|c|}
\hline
$m_{n-1}$ & $m_n$ & Weight & Meaning \\
\hline 
0 & 0 & $\me^{-\beta (t_n - t_{n-2})}$ & Direct wins at $t_{n-2}$ \\
\hline
0 & 1 & $\alpha \me^{-\beta (t_n - t_{n-2})}$ & Indirect wins (distance $2$) based on $t_{n-2}$, $t_n$ \\
\hline
1 & 0 & $\alpha \me^{-\beta (t_n - t_{n-2})}$ & Indirect wins (distance $2$) based on $t_{n-2}$, $t_{n-1}$ \\
\hline
1 & 1 & $\alpha^2 \me^{-\beta (t_n - t_{n-2})}$ & Indirect wins (distance $3$) based on $t_{n-2}$, $t_{n-1}$, $t_n$ \\
\hline
\end{tabular}
\end{center}
\end{table}

In this dynamic model, the order in which results occur influences whether indirect wins are awarded. If team A defeats team B and then team B defeats team C three years later, team A would not be credited with an indirect win. However, if team B had already recently defeated team C (likely meaning its quality was high at that time), then team A would receive an indirect win if it beat team B. This difference allows the dynamic model to capture the importance of a team's current quality during each match, rather than simply their quality at the end of the testing period. 

We can rearrange Equation \eqref{eq:dyn} to obtain an iterative expression for $W_{t_n}$:
\begin{align}\label{eq:dynsimp}
W_{t_n} &=  A_{t_n} + \me^{-\beta (t_n - t_{n-1})}  \mathlarger{\mathlarger{[}} A_{t_{n-1}} + \me^{-\beta (t_{n-1} - t_{n-2})}  \mathlarger{\sum}_{m_{n-1} \in \{0,1\}} \alpha^{m_{n-1}} A_{t_{n-2}} A_{t_{n-1}}^{m_{n-1}} + \cdots \nonumber \\
& + \me^{-\beta (t_{n-1} - t_{1})} 
\mathlarger{\sum}_{m_{2}, \ldots,~ m_{n-1} \in \{0,1\}} \alpha^{\sum_{i=2}^{n-1}  m_i} A_{t_{1}} A_{t_{2}}^{m_{2}} \cdots A_{t_{n-1}}^{m_{n-1}} \mathlarger{\mathlarger{]}} \mathlarger{\sum}_{m_n \in \{0,1\}} \alpha^{m_n} A_{t_n}^{m_n} \nonumber \\
& = A_{t_n} + \me^{-\beta (t_n - t_{n-1})} W_{t_{n-1}} (I + \alpha A_{t_n}).
\end{align}

This iterative expression leads to an update equation for the dynamic win score:
\begin{equation}
\bm{w}_{t_n} = 
\begin{cases}
A_{t_1}^T \bm{1},\qquad\qquad\qquad\qquad\qquad\qquad \qquad \text{ if } n=1, \\
A_{t_n}^T \bm{1} + \me^{-\beta (t_n - t_{n-1})} \left(I + \alpha A_{t_n}^T \right) \bm{w}_{t_{n-1}}, \qquad \text{if } n>1.
\end{cases}
\end{equation}

The dynamic loss score can be defined in the same way simply by replacing $A_{t_n}$ by $A_{t_n}^T$ (since this essentially switches wins to losses):
\begin{equation}
\bm{l}_{t_n} = 
\begin{cases}
A_{t_1} \bm{1},\qquad\qquad\qquad\qquad\qquad\qquad \qquad \hspace{0.3em}\text{ if } n=1, \\
A_{t_n} \bm{1} + \me^{-\beta (t_n - t_{n-1})} \left(I + \alpha A_{t_n} \right) \bm{l}_{t_{n-1}}, \quad\qquad \text{if } n>1.
\end{cases}
\end{equation}

As in the static model, we define the dynamic win-loss score at time $t_n$ in vector form as $\bm{s}_{t_n} = \bm{w}_{t_n} - \bm{l}_{t_n}$.
 
We make the stylistic choice in this work to redefine the decay parameter $\me^{-\beta (t_n - t_{n-1})}$ in terms of a more easily understandable parameter which will be interpreted in more detail in Section \ref{sec:param}. We define $\tau$ to be the \textit{bygone win factor}:
\begin{equation}
\tau = \me^{-\beta (365 ~\text{days})}.
\end{equation}
This definition symmetrises the decay rates for indirect wins and bygone wins, based on the following logic: each increase in distance gets discounted by a factor of $\alpha$, while each additional year in the past gets discounted by a factor of $\tau$. Although we could interpret the timescale of $\tau$ as days past, we believe that a yearly factor makes more intuitive sense and hence is more likely to be used in the real world. This model will allow conclusions to be drawn about the relative importance of indirect wins and bygone wins in terms of meaningful parameters, rather than the obscure $\beta$ used by Motegi and Masuda \cite{Motegi2012}.

\subsection{Parameter Estimation}

%Our primary test data in this section is the 2010 World Cup. One reason for this choice is that the 2014 World Cup took place after results on how to exploit the FIFA Rankings were released. In addition, the predictive accuracy for the FIFA Rankings in 2014 is a significant outlier from the previous five World Cups and should not be taken as normal. For these reasons, we choose the 2010 World Cup as it is the second most recent tournament and has a roughly average predictive accuracy: $67\%$.

In the Dynamic Model, we no longer have any limitation on $\alpha$ as the decay function in time automatically ensures the convergence of the series. (At each time $t_n$ the network is acyclic; there are no teams that play two matches on the same day.) Therefore, we consider all $\alpha \geq 0$ \cite{Motegi2012}. We consider only $\tau \in [0,\ 1]$ since we don't want the importance of bygone wins to exponentially increase. Figure \ref{fig:heat2010} shows the predictive accuracy of the Dynamic Model for the 2010 and 2014 World Cups with a range of $\alpha$ and $\tau$ values. We can clearly see that the highest predictive accuracy tends to arise when $\alpha = 0.07$. For $\tau \in [0.88,\ 0.96]$, the predictive accuracy with this $\alpha$ value is $79\%$ for both World Cups --- much higher than the average predictive accuracy of the FIFA Rankings.
\begin{figure}[H]
\begin{center}
\centering
\begin{subfigure}[t]{\textwidth}
\centering
\includegraphics[width=0.72\textwidth]{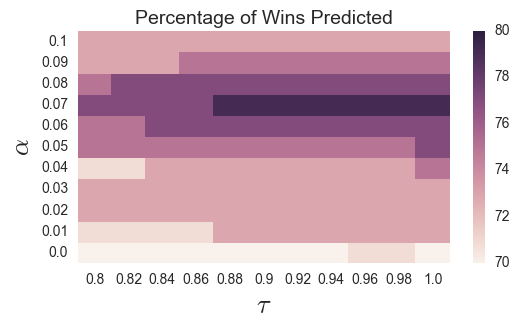}
\caption{2010 World Cup.}
\end{subfigure}
\begin{subfigure}[t]{\textwidth}
\centering
\includegraphics[width=0.72\textwidth]{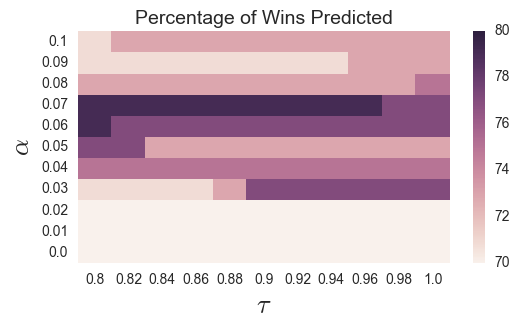}
\caption{2014 World Cup.}
\end{subfigure}
\end{center}
\vspace*{-5mm}
\caption{Investigation of the effect of $\alpha$ and $\tau$ on the predictive accuracy for the 2010 and 2014 World Cups using the Dynamic Model.}\label{fig:heat2010}
\end{figure}

We must determine whether this is indeed a significant improvement on the FIFA Rankings. To do so, we take $\tau = 0.9$ as this value falls inside the predictive accuracy maximising range. From here on, $\alpha$ is taken to be $0.07$. Table \ref{tab:preds} summarises the predictive accuracy values for the past three World Cups using these parameters.
\begin{table}[H]
\caption{Predictive accuracy of the Dynamic Model and FIFA Rankings for the past three World Cups, using $\alpha=0.07$ and $\tau=0.9$.}\label{tab:preds}
\vspace*{-5mm}
\begin{center}
\begin{tabular}{|c|c|c|}
\hline
World Cup & Dynamic Model accuracy (\%) & FIFA Rankings accuracy (\%) \\
\hline
2006 & 71 & 63 \\
2010 & 79 & 67 \\
2014 & 79 & 84 \\
\hline
Average & $76 \pm 5$ & $71 \pm 11$ \\
\hline
\end{tabular}
\end{center}
\end{table}
The results from this table are promising, if not conclusive; the Dynamic Model is at least as effective as the FIFA Rankings in predicting the results of World Cups. Having a lower standard deviation is also a desirable quality --- we want a model with high predictive accuracy, but also one which is consistent through time rather than fluctuating wildly like the FIFA Rankings.

\section{Analysis}\label{sec:pred}
\subsection{Predictive Accuracy}

We will now examine the predictive accuracy for all matches (not just World Cups). To do so, we calculate the rank of each team at each $t_n$ using $\alpha=0.07$ and $\tau = 0.9$. If a match result at $t_{n}$ is inconsistent with the prediction of the ranking from $t_{n-1}$, then we say that an \textit{upset} occurred at $t_{n}$. The predictive accuracy at $t_{n}$ is defined as
\begin{equation}
P_{t_{n}} = \frac{N_{t_{n}} - u_{t_{n}}}{N_{t_{n}}},
\end{equation}
where ${N_{t_{n}}}$ is the number of matches with a winner from $t_1$ to $t_n$ and $u_{t_n}$ is the number of upsets from $t_1$ to $t_n$. Figure \ref{fig:retro} shows that the predictive accuracy fluctuates around  $72$--$73\%$. The fact that this predictive accuracy is very near the $76\%$ World Cup predictive accuracy allows us to hypothesise that one can expect this level of accuracy from the Dynamic Model in general.
\begin{figure}[H]
\begin{center}
\includegraphics[width=0.75\textwidth]{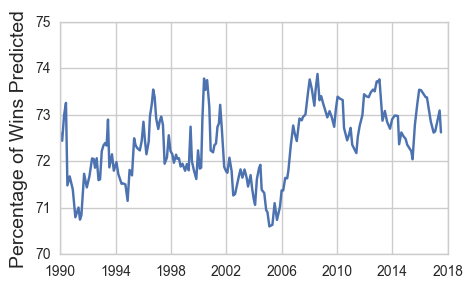}
\end{center}
\vspace*{-5mm}
\caption{Monthly predictive accuracy of the Dynamic Model with $\alpha=0.07$ and $\tau=0.9$.}\label{fig:retro}
\end{figure}

\subsection{Parameter Interpretation}\label{sec:param}

Parameter estimation for the Dynamic Model allows us to maximise the predictive accuracy and also leads to conclusions about the parameters themselves. Since we have purposely constructed a model whose parameters have real-world meaning, we can now interpret this meaning. 

Based on the results of Section \ref{sec:dynamic}, we determined that $\alpha = 0.07$ was the optimal value to maximize predictive accuracy. This means that an indirect win is worth $7\%$ of a direct win towards a team's win score. So to all football fans discussing the merit of the `team A beat team B who beat team C' argument, it can now be defended with a mathematical model, rather than simply more beer.

The other parameter, $\tau$, is even more interesting. It was shown that the optimal $\tau$ for World Cup predictive accuracy is roughly $0.9$. This suggests that a win a year ago is worth $90\%$ of a win yesterday, while a win two years ago is worth only $81\%$, and so on. This measure of relative importance is strongly at odds with the block-wise depreciation of the current FIFA Rankings. As can be seen in Figure \ref{fig:time}, the Dynamic Model suggests that bygone wins are much more important for predictive accuracy than FIFA gives them credit for. This is reasonably unsurprising; the results of the past two World Cups are both likely to still have some importance.
\begin{figure}[H]
\begin{center}
\includegraphics[width=0.75\textwidth]{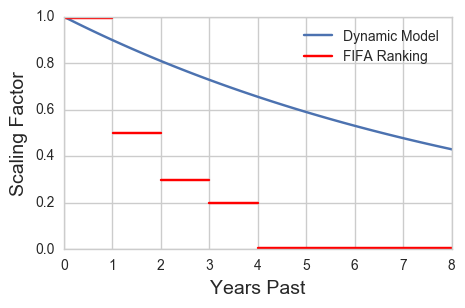}
\end{center}
\vspace*{-5mm}
\caption{Comparison of temporal scaling factors with $\tau=0.9$.}\label{fig:time}
\end{figure}

\subsection{Ranking Evolution}

Using the Dynamic Model, with $\alpha=0.07$ and $\tau=0.9$, we can now examine the rank of specific teams. Figure \ref{fig:teams} shows the monthly Dynamic Model rank for three traditionally strong teams. First, we note that the Dynamic Model is not overly sensitive to recent results; Germany, for example, stays within the top eight teams in the world for the past seven years. This model also accurately captures the precipitous end of Spain's reign at the top of world football near 2014, a finding that many ranking systems agree with.
\begin{figure}[H]
\begin{center}
\includegraphics[width=0.75\textwidth]{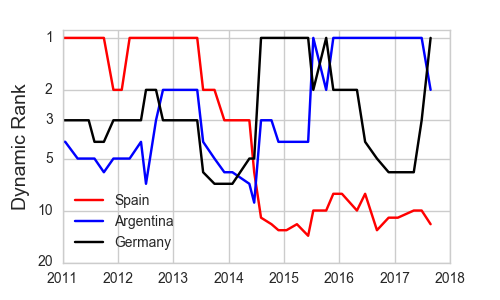}
\end{center}
\vspace*{-5mm}
\caption{Monthly Dynamic Model ranking for three teams using $\alpha=0.07$ and $\tau=0.9$.}\label{fig:teams}
\end{figure}

\section{Discussion}\label{sec:disc}

Having presented two models, it is clear that the Dynamic Model outperforms the Static Model. There is a tradeoff between simplicity and predictive accuracy between the two, and we acknowledge that predictive accuracy should not be the only metric by which ranking systems are measured. If it were, teams should simply be ranked by bookmakers or by complex algorithms like ESPN's Soccer Power Index which takes into account each individual player's current club performance \cite{spi}. Complexity simply for the sake of predictive accuracy has drawbacks. A lack of transparency or comprehensibility can hinder its adoption --- especially with an organisation such as FIFA that, mired in recent controversy over corruption, likely wishes to be as candid as possible. We believe that the Dynamic Model, while more complex than its static counterpart, still fits this mould. It has only two (intuitive) parameters, it reduces continental bias, and it removes the obvious method of exploitation. The dependence on $\lambda_{\text{max}}$ from the Static Model is also removed, meaning that the parameters are independent of the adjacency matrix at any point in time. Finally, the dynamic nature means that this model can capture the temporal fluctuations in team quality, a feature that results in roughly $10\%$ higher predictive accuracy than the Static Model. %For these reasons, we suggest that FIFA should adopt the Dynamic Model as its future ranking system.
%73 vs 76 -- normally good teams are predicted better than bad

One apparent issue with our metric of predictive accuracy is that it is taken in October of the preceding year --- would this make all of our predictions outdated? Almost certainly, to an extent. If a team suddenly gets much better or worse, then presumably our model will not have predicted this all the way back in October. To examine the influence that the time of prediction has on the predictive accuracy, we carried out a similar analysis using the Dynamic Model rankings the day before a World Cup started. For the past three World Cups, the day-before predictions were one match ($\approx 2\%$) better than the October predictions, on average. This improvement is expected, but its small size suggests that the upsets which occur during a World Cup really are upsets and do not simply reflect outdated rankings. 

The three major critiques of the FIFA Rankings are that their predictive accuracy is low, they contain continental bias, and they can be exploited by careful selection of which teams should be played when. The Dynamic Model clearly remedies the first of these critiques, with an improvement of $5\%$ in World Cup predictive accuracy. It also replaces the team and continent multipliers with the concept of indirect wins, removing the continental bias. Finally, it removes the obvious method of exploitation of the FIFA Rankings where teams can improve their average point score by playing fewer friendly matches. However, this work has not examined how to exploit the Dynamic Model. Crucial future work could involve analysing the upsets: are they primarily from one continent? Is our model exploitable by any scheduling method? Could a team optimise their rank by taking strategic breaks after large upsets to not risk falling in the ranks? All of these questions are valid, and would make for fascinating future work. However, the primary purpose of this work was to address the main critiques of the FIFA Rankings and to provide a justifiable alternative: the network-based Dynamic Model.

\subsection*{Acknowledgements}

Thank you to Renaud Lambiotte for his helpful comments on a draft of this work.

\newpage
\small
%\pagenumbering{gobble}
%\Urlmuskip=0mu plus 1mu
%\bibliographystyle{siamplain}
%\bibliography{/home/sam/Desktop/Oxford/BibTex/FootballNetworks}
%\bibliography{FootballNetworks}

\end{document}